\documentclass[a4paper,UKenglish,cleveref, autoref, thm-restate]{lipics-v2021}
\usepackage{mathtools}
\usepackage{amsmath} 
\usepackage{todonotes} 
\usepackage{nicefrac} 
\usepackage{tikz}
\usetikzlibrary{arrows,arrows.meta,decorations.pathmorphing,decorations.pathreplacing,decorations.shapes,decorations.markings,decorations,backgrounds,positioning,patterns,calc,fit,matrix,shapes,petri,topaths,intersections,decorations.pathmorphing}

\hideLIPIcs
\nolinenumbers

\tikzset{
    vertex/.style={circle,draw,fill=black, inner sep= 1.5pt},
} 
\tikzstyle{vert2}=[circle,inner sep=1.5,fill=white,draw,minimum size=.2cm]
\tikzstyle{square}=[regular polygon,regular polygon sides=4]
\tikzstyle{edge}=[thick]

\title{Restless Temporal Path Parameterized Above Lower Bounds} %

\author{Philipp Zschoche}{Algorithmics and Computational Complexity, Faculty IV, Technische Universität Berlin, Germany}{zschoche@tu-berlin.de}{https://orcid.org/0000-0001-9846-0600}{}

\authorrunning{Philipp Zschoche} %

\Copyright{Philipp Zschoche} %



\category{} %

\relatedversion{} %
\EventEditors{John Q. Open and Joan R. Access}
\EventNoEds{2}
\EventLongTitle{42nd Conference on Very Important Topics (CVIT 2016)}
\EventShortTitle{CVIT 2016}
\EventAcronym{CVIT}
\EventYear{2016}
\EventDate{December 24--27, 2016}
\EventLocation{Little Whinging, United Kingdom}
\EventLogo{}
\SeriesVolume{42}
\ArticleNo{23}

\newcommand{\probSRestlessPath}{\textsc{Short Restless Temporal Path}}
\newcommand{\FPT}{\textsc{FPT}}
\newcommand{\NP}{\textsc{NP}}
\newcommand{\XP}{\textsc{XP}}
\newcommand{\Wone}{\textsc{W[1]}}
\newcommand{\PequalNP}{\textsc{P}$=$\textsc{NP}}
\newcommand{\yes}{\emph{yes}}
\newcommand{\no}{\emph{no}}

\newcommand{\TG}{\mathcal G}
\newcommand{\TE}{\mathcal E}
\newcommand{\lifetime}{\tau}
\newcommand{\TGcompact}{\TG := (V,(E_i)_{i=1}^\lifetime)}
\newcommand{\NN}{\mathbb N}
\newcommand{\wait}{\delta}
\newcommand{\z}{z}
\newcommand{\etal}{et~al.}

\DeclarePairedDelimiterX{\abs}[1]{\lvert}{\rvert}{#1}
\newcommand{\set}[1]{[#1]}
\newcommand{\defeq}{:=}%
\newcommand{\mvert}{\;\middle\vert\;}
\newcommand{\decprob}[3]{
	\begin{center}
		\begin{minipage}{0.97\linewidth}%
			\noindent
			\textsc{#1}
			\begin{description}
			\item[\textcolor{black}{\textbf{Input:}}]  #2
			\item[\textcolor{black}{\textbf{Question:}}]  #3
			\end{description}
		\end{minipage}
	\end{center}
}

\begin{document}

\maketitle

\newcommand{\runningtimeInL}{O(4^{\ell} \cdot \ell^2 \abs{\TG}^3\wait\log(\nicefrac{k}{p\ell}))}
\begin{abstract}
Reachability questions are one of the most fundamental algorithmic primitives in temporal 
graphs---graphs whose edge set changes over discrete time steps.
A core problem here is  the \NP-hard \probSRestlessPath{}:
given a temporal graph $\TG$, two distinct vertices $s$ and $\z$, and
two numbers $\wait$ and $k$,
is there a $\wait$-restless temporal $s$-$\z$~path of length at most~$k$?
A temporal path is 
a path whose edges appear 
in chronological order and a temporal path is $\wait$-restless if two consecutive path edges appear at most~$\wait$ time steps apart from each other.
Among others, this problem has applications in neuroscience and epidemiology.
While \probSRestlessPath{} is known to be computationally hard, 
e.g.,~it is \NP-hard for only three time steps and \Wone-hard when parameterized by the feedback vertex number of the underlying graph,
it is fixed-parameter tractable when parameterized by the path length~$k$.
We improve on this by showing that \probSRestlessPath{} can be solved in (randomized) $4^{k-d}\abs{\TG}^{O(1)}$~time,
where~$d$ is the minimum length of a temporal $s$-$\z$~path.
\end{abstract}

\section{Introduction}
Susceptible-Infected-Recovered. 
These are the three states of the \emph{SIR-model}---a 
canonical spreading model for diseases where recovery confers lasting resistance~\cite{Bar16,kermack1927contribution,New18}.
Here, an individual is at first susceptible (S) to get a certain disease, can devolve to be infected (I), 
and ends up resilient 
after recovery (R).
We study one of the most fundamental algorithmic questions in this model:
given a set of individuals with a list of physical contacts over time, and two individuals~$s$ and~$\z$, 
is it possible to have a chain of infections from~$s$ to~$\z$?
As the timing of the physical contacts is crucial in this scenario,
we use a \emph{temporal graph} $\TGcompact$ consisting of a set~$V$ of vertices
and an edge set that changes over discrete time steps described by 
a chronologically ordered sequence~$(E_i)_{i=1}^\lifetime$ of edge sets over~$V$.
A temporal path is a path whose edges appear in chronological order.
In particular, 
a sequence $P := ((e_i,t_i))_{i=1}^m$ of time-edges from  $\TE(\TG) := \bigcup_{i=1}^\lifetime E_i \times \{i\}$
is a \emph{temporal $s$-$\z$~path} of length~$m$
if $(\bigcup_{i=1}^m e_i,\{ e_i \mid i \in \set{m} \})$ is an $s$-$\z$~path (no vertex is visited twice)
and $t_i \leq t_{i+1}$ for all~$i \in \set{m-1}$.
If we construct a temporal graph where the vertices are individuals and an edge $e \in E_t$ represents a physical contact of two individuals
at time step~$t$, then a chain of infections is represented by a temporal path.
However, not every temporal path yields a potential chain of infections, as an infected person might recover before the next individual is met. 
To represent infection chains in the SIR-model by temporal paths,
we restrict the waiting time at each intermediate vertex to a prescribed duration---that is, 
the time until an individual becomes resilient after infection.
These temporal paths are called restless.
In particular, the temporal $s$-$\z$~path $P$ is \emph{$\wait$-restless} 
if~$t_i \leq t_{i+1} \leq t_i + \wait$ for all~$i \in \set{m-1}$.
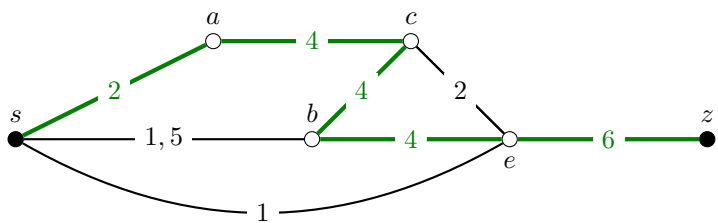
\begin{figure}[t]
  \centering
\begin{tikzpicture}[scale=1.3,]
		\node[vert2,fill=black,label=$s$] (s) at (0,0) {};
		\node[vert2,label=$a$] (a) at (2,1) {};
		\node[vert2,label={below:$e$}] (b) at (5,0) {};
		\node[vert2,label=$c$] (c) at (4,1) {};
		\node[vert2,label=$b$] (d) at (3,0) {};
		\node[vert2,fill=black,label=$\z$] (z) at (7,0) {};

		\draw[edge,green!50!black, ultra thick] (s) -- node[fill=white] {$2$} (a);
		\draw[edge] (s) to [bend right]  node[fill=white] {$1$}  (b);
		\draw[edge,green!50!black, ultra thick] (a) -- node[fill=white] {$4$} (c);
		\draw[edge] (b) -- node[fill=white] {$2$} (c);
		\draw[edge] (b) -- node[fill=white] {$2$} (c);
		\draw[edge,green!50!black, ultra thick] (c) -- node[fill=white] {$4$} (d);
		\draw[edge,green!50!black, ultra thick] (d) to node[fill=white] {$4$} (b);
		\draw[edge,green!50!black, ultra thick] (b) to node[fill=white] {$6$} (z);
		\draw[edge] (s) to node[fill=white] {$1,5$} (d);

\end{tikzpicture}
\caption{An illustration of a temporal graph with vertices $s,a,b,c,d,e$, and~$\z$.
		The labels on the edges denote at which time steps the edges are present.
		The time-edges of a~$2$-restless temporal $s$-$\z$~path in this temporal graph are
		marked by thick (green) edges. 
		In fact, this is the only~$2$-restless temporal~$s$-$\z$~path in this temporal graph,
		as we cannot visit a vertex twice and two consecutive time-edges have to be at most two time steps apart.
}
\label{fig:simple-example}
\end{figure}
Hence, restless temporal paths model infection transmission routes of diseases that grant immunity upon recovery~\cite{Hol16}.
Other applications of restless temporal paths 
appear in the context of delay-tolerant networking with time-aware routing tables \cite{CasteigtsHMZ20},
and in the context of
finding signaling pathways in brain networks \cite{thejaswi2020restless}.
Consider \cref{fig:simple-example} for an illustration of a temporal graph 
with a $2$-restless temporal $s$-$\z$~path.

The central problem of this work is as follows. %
\decprob{\probSRestlessPath}
{A temporal graph $\TG$, a source vertex $s \in V$, a destination vertex $\z \in V$, and integers $\wait,k \in \mathbb N$.}
{Is there a $\wait$-restless temporal $s$-$\z$~path in~$\TG$ of length at most~$k$?}

Casteigts~\etal~\cite{CasteigtsHMZ20} showed that \probSRestlessPath{} is \NP-hard even 
if~$\wait=1$, 
		$\lifetime=3$, 
		every edge appears only once, and
		the underlying graph has a maximum degree of six.
		Moreover, they showed that it is \Wone-hard when parameterized by the distance to disjoint paths of the underlying 
		graph\footnote{That is, the minimum number of vertices we need to remove from a graph such that the remaining graph consists of a set of vertex-disjoint paths.}.

		Hence, \probSRestlessPath{} is presumably not fixed-parameter tractable when parameterized
		by a wide range of well-known parameters of the underlying (static) graph, e.g., 
		feedback vertex number, pathwidth, or cliquewith.
		However, \probSRestlessPath{} is fixed-parameter tractable when parameterized by~$k$
		or the treedepth of the underlying graph or the feedback edge number of the underlying graph~\cite{CasteigtsHMZ20}.
		 Thejaswi~\etal~\cite{thejaswi2020restless} showed 
		that for every $p \in \mathbb R$ with $0<p<1$ 
		there is a randomized~$O(2^kk \abs{\TG}\wait\log(k\cdot \nicefrac{1}{p}))$-time 
		algorithm for \probSRestlessPath{} that has a one-sided
		error probability of at most~$p$.
		More precisely, if the algorithm returns \yes{}, then the given instance $I$ of \probSRestlessPath{} is a \yes-instance,
		and if the algorithm returns \no{}, then the probability that $I$ is a \yes-instance is at most~$p$.
		They conducted experiments on large synthetic and real-world data sets and showed that their algorithm performs 
		well as long as the parameter~$k$ is small. 
		For example, one can solve \probSRestlessPath{} with $k \leq 9$
		and a temporal graph with $36$~million time-edges in less than one hour with customary desktop hardware.
		On the data set used in the experiments, $k$ seems to be the only useful parameter for which we know 
		that \probSRestlessPath{} is fixed-parameter tractable;
		all other known parameters 
		(i.e., timed feedback vertex number~\cite{CasteigtsHMZ20}, 
		treedepth of the underlying graph, and feedback edge number of the underlying graph) 
		are too large to be eligible in practice~\cite{thejaswi2020restless}.
		Hence, the current algorithms are not satisfactory when it comes to computing long restless temporal paths 
		in real-world temporal networks.

		The parameter~$k$ of \probSRestlessPath{} can be seen as the \emph{solution size} and is thus a natural and
		well-motivated parameter from a parameterized algorithmics point of view.
		However, as we observed before, \FPT-algorithms regarding the solution size are not necessarily practical,
		e.g., if all solutions are large.
		To address this problem, one can investigate 
		\emph{parameterizations above guaranteed lower bounds} 
		\cite{%
DBLP:journals/algorithmica/AlonGKSY11,%
DBLP:journals/siamdm/BezakovaCDF19,%
DBLP:journals/jcss/CrowstonFGJKRRTY14,%
DBLP:journals/mst/GutinKLM11,%
DBLP:journals/jcss/GutinIMY12,%
DBLP:journals/jal/MahajanR99,%
DBLP:journals/jcss/MahajanRS09}:
that is, the \emph{difference} between the smallest size of a solution and a guaranteed lower bound for the solution size.
In the case of \probSRestlessPath{}, three lower-bounds for $k$ seem
particularly interesting:
\begin{description}

		\item[The distance from $s$ to $\z$:] 
				The minimum length of an $s$-$\z$~path in the underlying graph. 
		\item[The temporal distance from $s$ to $\z$:]
				The minimum length of a temporal $s$-$\z$~path.
		\item[The $\wait$-restless temporal distance from $s$ to $\z$:]
				The minimum length of a $\wait$-restless temporal $s$-$\z$~walk.
				Herein,
a sequence $W := ((e_i,t_i))_{i=1}^m$ of time-edges 
is a \emph{temporal $s$-$\z$~walk} of length~$m$
if the edges $(e_i)_{i=1}^m$ induce an $s$-$\z$~walk
and $t_i \leq t_{i+1}$ for all~$i \in \set{m-1}$.
Moreover, $W$ is $\wait$-restless if $m=1$ or 
$t_{i+1} - t_i \leq \wait$.
\end{description}

Note that the length of a $\wait$-restless temporal $s$-$\z$~path 
is at least 
the minimum length of a $\wait$-restless temporal $s$-$\z$~walk 
which is in turn at least 
the minimum length of a temporal $s$-$\z$~path which is again at least
the minimum length of an $s$-$\z$~path in the underlying graph.
For the sake of brevity, we say for an 
instance~$(\TG,s,\z,\wait,k)$ of \probSRestlessPath{}
that 
the $\wait$-restless temporal distance from $s$ to $\z$,
the temporal distance from $s$ to $\z$, or
the distance from $s$ to $\z$ is $k+1$ 
if there is no 
$\wait$-restless temporal $s$-$\z$~walk,
no temporal $s$-$\z$~path, or 
no $s$-$\z$~path in the underlying graph, respectively.

Unfortunately, 
a closer look at the \NP-hardness reductions of Casteigts~\etal~\cite{CasteigtsHMZ20} reveals that, unless~\PequalNP{},
there is not even a $\abs{\TG}^{f(k-d_r)}$-time algorithm for \probSRestlessPath{},
where $d_r$ is the $\wait$-restless temporal distance from $s$ to $\z$ and $f$ is a computable function. 

\subparagraph{Our contributions.}
We show that \probSRestlessPath{} can be solved in randomized $4^{k-d}\abs{\TG}^{O(1)}$~time,
where $d$ is the temporal distance from~$s$ to~$\z$.
To the best of our knowledge, this is the first above-lower-bound \FPT-algorithm on temporal graphs.
More precisely, we show that for every~$p \in \mathbb R$ with~$0 < p < 1$ 
there is a randomized~$\runningtimeInL$-time algorithm for \probSRestlessPath{}
with a one-sided error probability of at most~$p$, 
where $\ell:=k-d$ and~$d$ is the temporal distance from~$s$ to~$z$.
The main technical contribution behind this is a geometrical perspective onto temporal graphs which seems applicable to other temporal graph problems when parameterized above the temporal distance between vertices.
In the resulting algorithm, 
the only subroutine with a super-polynomial running time
is the algorithm of Thejaswi~\etal~\cite{thejaswi2020restless} 
that we employ to find $\wait$-restless temporal path of length at most $2(k-d)+1$.
In fact, 
this subroutine can be replaced by a deterministic algorithm of  Casteigts~\etal~\cite{CasteigtsHMZ20}---this leads 
to a $2^{O(k-d)}\abs{\TG}^3\wait$-time deterministic algorithm 
for \probSRestlessPath{}.
The running time overhead induced by our technique is $O(\abs{\TG}^2\ell)$ in the deterministic case 
and $O(\abs{\TG}^2\ell\log(\nicefrac{k}{\ell p}))$ if we use the algorithm of Thejaswi~\etal~\cite{thejaswi2020restless},
where $\ell := k - d$ and $d$ is the temporal distance from~$s$ to~$z$.
The overhead with the randomized algorithm is larger as we need that 
the error probability of several calls of 
the randomized algorithm accumulate to~$p$.
Although the running time overhead of our technique is is slightly larger with the randomized algorithm of Thejaswi~\etal~\cite{thejaswi2020restless} because a faster overall running time.

\subparagraph{Further related work.}
In the literature, waiting time constraints are studied from various angles.
Himmel~\etal~\cite{himmel_efficient_2019}
studied a variant of restless temporal paths where multiple
visits of vertices are permitted, i.e., restless temporal
\emph{walks}. 
In contrast to restless temporal paths, they showed that such walks can be 
computed in polynomial time. %
Pan and Saramäki~\cite{PS11} empirically studied the correlation between 
waiting times of temporal paths and the ratio of
the network reached in spreading processes.
Akrida~\etal~\cite{Akr+19a} studied flows in temporal networks with ``vertex buffers'', 
which however pertains to the quantity of
information that a vertex can store, rather than a duration.

Algorithmic reachability questions are one of the most thriving research topics in temporal graphs.
Bui{-}Xuan~\etal~\cite{XuanFJ03} and Wu~\etal~\cite{wu_efficient_2016} studied the computation of temporal paths that
satisfy certain optimality criteria and show that shortest, fastest, and foremost temporal path can be computed in 
polynomial time. 
In the temporal setting, reachability is not an
equivalence relation among vertices 
and the reachability relation between vertices is not even transitive---this makes many problems computationally harder than 
their counterpart on static graphs. 
Michail~and~Spirakis~\cite{michail2016traveling} studied the \NP-hard question of whether a temporal graph contains a 
temporal walk that visits each vertex at least once.
This problem remains computationally hard 
even if the underlying graph is a star \cite{AMS19,BM21}.
If the underlying graph is connected at each time step
and the walk can only contain one edge in each time step,
then a fast exploration is guaranteed \cite{Erlebach0K15,ErlebachS18,ErlebachKLSS19}.
However, on these so-called always-connected temporal graphs, the decision problem remains \NP-hard, even if the underlying graph has pathwidth two \cite{BodlaenderZ19}.
Kempe~\etal~\cite{kempe2002connectivity} studied
whether there are $k$ vertex-disjoint temporal paths between two given vertices. 
While the classical analogue of this on static graphs is polynomial-time solvable, it becomes \NP-hard in the temporal setting.
Moreover, this problem remains \NP-hard on a single underlying path, when we are looking for a set of temporal paths which is only 
pairwise vertex-disjoint at any point in time \cite{KMMNZ21}.
Furthermore, the related problem of finding small separators in temporal graphs becomes computationally hard \cite{FluschnikMNRZ20,kempe2002connectivity}, even on quite restricted temporal graph classes~\cite{FluschnikMNRZ20}.
Bhadra and Ferreira~\cite{BF03} showed that finding a maximum temporally connected component is \NP-hard. 
Furthermore, a temporal graph may not have a 
sparse spanner~\cite{AxiotisF16},
and computing a spanner 
with a minimum number of time-edges is \NP-hard~\cite{Akr+17,mertzios2019temporal}.

Related to spreading processes, 
Enright~\etal~\cite{EnrightMMZ19,enright2021assigning}, 
Deligkas and Potapov~\cite{DeligkasP20},
and Molter~\etal~\cite{MRZ21} 
studied restricting the set of 
reachable vertices via various 
temporal graph modifications---all described decision problems 
are \NP-hard in rather restricted settings.

\section{Preliminaries} 
We denote by~$\NN$ and~$\NN_0$ the natural numbers excluding and including zero, respectively.
By $\mathbb R$,~$\mathbb Q$, and~$\mathbb Z$ we denote the real numbers, rational numbers, and the integers, respectively.
Moreover, 
$[a,b] \defeq \{i \in \mathbb Z \mid a\leq i \leq b\}$, $[n]\defeq[1,n]$, 
$\mathbb R_+ \defeq \{ x \in \mathbb R \mid x \geq 0 \}$,
and~$\mathbb Q_+ \defeq \{ x \in \mathbb Q \mid x \geq 0 \}$.
We denote by $\log(x)$ the ceiling of the binary logarithm of $x$ ($\lceil \log_2(x)\rceil$), 
where~$x \in \mathbb R$.

Let $(a_i)_{i=1}^n \defeq (a_1,a_2,\dots,a_n)$ be a sequence of length $n$ 
and let $(b_i)_{i=1}^m$ be a sequence of length $m$.
We denote by~$x \in (a_i)_{i=1}^n$ that there is an~$i \in \set{n}$ such that~$x=a_i$.
We denote by~$(a_i)_{i=1}^n \subseteq (b_i)_{i=1}^m$ that~$(a_i)_{i=1}^n$ is a \emph{subseqence} of~$(b_i)_{i=1}^m$.
That is, there is an injective function $\sigma \colon \set{n} \to \set{m}$ 
such that $a_i = b_{\sigma(i)}$ for all $i \in \set{n}$ 
and $\sigma(i) < \sigma(j)$ for all $i,j \in \set{n}$ with $i<j$.
Moreover, for a set~$S$,
we denote by $(a_i)_{i=1}^n \setminus S$ the subsequence of $(a_i)_{i=1}^n$ 
where an element~$a_i$ is removed if and only if~$a_i \in S$, for all~$i \in \set{n}$.
Appending an element~$x$ to sequence~$(a_i)_{i=1}^n$ results in the sequence~$(a_i)_{i=1}^{n+1}$, where~$a_{n+1} = x$.

A \emph{randomized (Monte-Carlo) algorithm} has additionally access to an 
oracle that, given some number~$n \in \NN$, draws a value~$x \in \set{n}$ uniformly at random in constant time.
A (randomized) algorithm with error probability~$p$ 
is a randomized algorithm that returns the correct answer with probability~$1-p$.
For a finite alphabet~$\Sigma$ and a language $L \subseteq \Sigma^*$,
a (randomized) algorithm for~$L$ with a one-sided error probability~$p$
is a randomized algorithm that returns for every input $x \in \Sigma^*$ either \yes{} or \no{},
and one of the following is true:
\begin{itemize}
		\item If \yes{} is returned, then~$x \in L$ with probability~$1$.
				If \no{} is returned, then~$x \in L$ with probability~$p$.
		\item If \yes{} is returned, then~$x \not\in L$ with probability~$p$.
				If \no{} is returned, then~$x \not\in L$ with probability~$1$.
\end{itemize}
We refer to Mitzenmacher and Upfal~\cite{DBLP:books/daglib/0012859} for more material on randomized algorithms.
If it is not stated otherwise, then we use standard notation from graph theory \cite{Die16}.
Graphs are simple and undirected by default.

\subparagraph{Temporal graphs.}
A temporal graph $\TGcompact$ consists of a set of vertices $V(\TG) \defeq V$ and a sequence of edge sets $(E_i)_{i=1}^\lifetime$.
The number $\lifetime$ is the \emph{lifetime} of $\TG$.
The elements of $\TE(\TG) \defeq \bigcup_{i \in [\lifetime]} E_i \times \{i\}$ are called the 
\emph{time-edges} of $\TG$.
We say that time-edge $(e,t) \in  \TE(\TG)$ has time stamp $t$ and is in time step $t$.
The graph $(V,E_i)$ is called \emph{layer $i$} of temporal graph $\TG$, for all $i \in \set{\lifetime}$.
The \emph{underlying graph} of~$\TG$ is the (static) graph $(V,\bigcup_{i=1}^\lifetime E_i)$.
For every $v\in V$ and every $t\in [\lifetime]$, 
we denote the \emph{appearance of vertex}~$v$ \emph{at time}~$t$ 
by the pair $(v,t)$. 
For a time-edge $(\{v,w\},t)$ we call the vertex appearances $(v,t)$ and $(w,t)$ its \emph{endpoints}.
We assume that the \emph{size} 
of $\TG$ is~$\abs{\TG}\defeq\abs{V}+\sum_{i=1}^\lifetime \max\{1,\abs{E_i}\}$, 
that is, we do not assume to have 
compact representations of temporal graphs. 
For a vertex set $X \subseteq V$ of a temporal graph $\TGcompact$, 
we denote by $\TG[X]$ the temporal graph $(X,(E'_i)_{i=1}^\lifetime)$, 
where~$E'_i \defeq \left\{ e \in E_i \mvert e \subseteq X \right\}$.
Moreover, we denote the temporal graph $\TG$ 
without the vertices $X$ by $\TG - X \defeq \TG[V \setminus X]$.
For a time-edge set $Y$, 
we denote by $\TG \setminus Y$ the temporal graph 
where~$V(\TG\setminus Y) \defeq V(\TG)$ 
and~$\TE(\TG\setminus Y) \defeq \TE(\TG) \setminus Y$.

The set of vertices of the temporal path~$P = (e_i = (\{v_{i-1},v_i\},t_i))_{i=1}^m$ is denoted by~$V(P)=\left\{v_i \mvert i \in [m] \cup \{0\}\right\}$.
We say that~$P$ \emph{visits} the vertex~$v_i$ at time~$t$ if~$t \in [t_i,t_{i+1}]$, where~$i \in [m-1]$.
The \emph{departure} (or \emph{starting}) \emph{time}
of~$P$ is~$t_1$ and the \emph{arrival time} of~$P$ is~$t_m$.
A ($\wait$-restless) temporal $s$-$\z$~path of length~$m$ in a temporal graph~$\TG$ is a \emph{shortest} ($\wait$-restless) temporal $s$-$\z$~path 
if each temporal $s$-$\z$~path in~$\TG$ is of length at least~$m$.

A \emph{solution} of an instance~$(\TG,s,z,\wait,k)$ of \probSRestlessPath{} is a $\wait$-restless temporal $s$-$\z$ path of length at most $k$ in~$\TG$.

\subparagraph{Parameterized complexity.}
Let~$\Sigma$ be a finite alphabet.
A \emph{parameterized problem}~$L$ is a subset~$L\subseteq \Sigma^*\times \NN_0$.
The size of an instance $(x,k)\in\Sigma^*\times \NN_0$ is denoted by~$\abs{x}$ and usually we have that $\abs{x}+k \in O(\abs{x})$.
An instance~$(x,k)\in\Sigma^*\times \NN_0$ is a \yes-instance of~$L$ if and only if~$(x,k)\in L$ (otherwise it is a \no-instance).
A parameterized problem~$L$ is \emph{fixed-parameter tractable} (in \FPT) 
if
there is an (\FPT-)algorithm that decides
for every input~$(x,k)\in\Sigma^*\times \NN_0$ in~$f(k)\cdot |x|^{O(1)}$~time 
whether~$(x,k)\in L$, 
where~$f$ is some computable function only depending on~$k$.  
By slightly abusing the \FPT-terminology, we sometimes say that a parameterized problem is fixed-parameter tractable
even if the \FPT-algorithm has a constant one-sided error probability.
A parameterized problem~$L$ is in \XP{}
if for every input~$(x,k)$ one can decide in~$|x|^{f(k)}$~time whether~$(x,k)\in L$, 
where~$f$ is some computable function only depending on~$k$.  

The parameterized analogous of \NP{} and \NP-hardness are the \textsc{W}-hierarchy 
\begin{align*}
		\ \ \ \ \ \ \ \ \ \ \ \
		\ \ \ \ \ \ \ \ \ \ \ \
		\text{		\FPT{} $\subseteq$ \Wone{} $\subseteq$ \textsc{W[2]} $\subseteq \dots \subseteq$ \textsc{W[P]} $\subseteq$ \XP{} }
\end{align*}
and \textsc{W[}t\textsc{]}-hardness, 
where $t \in \NN \cup \{ $P$ \}$ and all inclusions are conjectured to be strict.
If some \textsc{W[}t\textsc{]}-hard parameterized problem is in \FPT, 
then \textsc{FPT}$=$\textsc{W[}t\textsc{]}.
We refer to 
Flum and Grohe~\cite{DBLP:series/txtcs/FlumG06},
Downey and Fellows~\cite{DF13}, and
Cygan~\etal~\cite{cygan2015parameterized}
for more material on parameterized complexity.

\section{The Algorithm} 
\label{sec:restless:above-lb}
In this section, we show that \probSRestlessPath{} can be solved in~$4^{k-d}\cdot \abs{\TG}^{O(1)}$~time with a constant one-sided error probability, where~$d$ is the minimum length of a temporal $s$-$\z$~path.
More precisely, we show the following.
\newcommand{\mainthm}{For every~$p \in \mathbb R$ with~$0 < p < 1$,
		there is a randomized $\runningtimeInL$-time algorithm for \probSRestlessPath{},
		where~$\ell := k-d$ and~$d$ is the minimum length of a temporal~$s$-$\z$~path.
		If this algorithm returns \yes, then the given instance is a \yes-instance.
		If this algorithm returns \no, 
then with probability of at least~$1-p$ the given instance is a \no-instance.}%
\begin{theorem}
		\label{thm:restless-abl}
		\mainthm
\end{theorem}

The proof of \cref{thm:restless-abl} is deferred to the end of this section.
In a nutshell, we use a prudent dynamic programming approach 
where we only check for $\wait$-restless temporal paths whose length is upper-bounded by $2(k-d)+1$ 
and then puzzle them together to ultimately find a $\wait$-restless temporal $s$-$z$~path,
where~$d$ is the minimum length of a temporal $s$-$\z$~path.
To detect $\wait$-restless temporal paths of some given length,
we employ the algorithm of Thejaswi~\etal~\cite{thejaswi2020restless}.
\begin{proposition}[\cite{thejaswi2020restless}]
		\label{lem:exact-restless-path}
For every $p \in \mathbb R$ with $0<p<1$ 
		there is a randomized~$O(2^k \cdot k \abs{\TG}\wait\log(k\cdot \nicefrac{1}{p}))$-time 
		algorithm that takes as input 
		a temporal graph~$\TG$,
		two vertices~$s,\z$, and 
		two integers~$\wait,k$.
		If the algorithm returns \yes,
		then there is a~$\wait$-restless temporal~$s$-$\z$~path
		of length exactly~$k$ in~$\TG$.
		If the algorithm returns \no,
		then with probability at least~$1-p$ there is no~$\wait$-restless temporal $s$-$\z$~path
		of length exactly~$k$ in~$\TG$.
\end{proposition}
In our algorithm, \cref{lem:exact-restless-path} can be replaced by any algorithm to find $\wait$-restless temporal paths of 
length~$k$.
For example, with the deterministic $2^{O(k)}\cdot \abs{\TG}\wait$-time algorithm of Casteigts~\etal~\cite{CasteigtsHMZ20} instead of \cref{lem:exact-restless-path}, we would end up with a
$2^{O(k-d)}\cdot \abs{\TG}^3\wait$-time algorithm for \probSRestlessPath{} that is deterministic. 
The precise running time overhead induced by our technique is $O(\abs{\TG}^2(k-d))$ time if we use a deterministic algorithm instead of \cref{lem:exact-restless-path} 
and $O(\abs{\TG}^2(k-d)\log(\nicefrac{k}{(k-d)p}))$ with \cref{lem:exact-restless-path}. 
The running time overhead with the randomized algorithm is larger as we need that 
the error probability of several calls of 
the randomized algorithm accumulate to~$p$.
Although the running time overhead of our technique is is slightly larger with the randomized algorithm of Thejaswi~\etal~\cite{thejaswi2020restless} because a faster overall running time.

For many algorithms based on dynamic programming, we have that the best-case running time
is not better than the worst-case running time.
In our case, we will realize that for sparse real-world graphs it seems that 
the caused overhead stays below the worst case.

In \cref{sec:general-idea}, we set up the geometric perspective on temporal graph
based on the temporal distance between vertices.
This might be of independent interest, as the ideas seem 
to be transferable to other problems where
an above-lower-bound parameterization by shortest temporal paths is possible.
In \cref{sec:dp}, we design a dynamic program to solve \probSRestlessPath{}
in $4^{k-d}\cdot \abs{\TG}^{O(1)}$ time, where~$d$ is the minimum length of a temporal $s$-$\z$~path.
In \cref{sec:final-poof}, we finally prove \cref{thm:restless-abl}.

\subsection{Geometric Perspective on Temporal Graphs Based on Shortest Temporal Paths}
\label{sec:general-idea}
In this section, we present the key ideas of the algorithm behind \cref{thm:restless-abl}.
To this end, we need some notation.
Let $\TGcompact$ be a temporal graph with two distinct vertices $s,z \in V(\TG)$ and $\wait,k \in \NN$.
We define the distance function~$d_\TG \colon V(\TG) \times \set{\lifetime} \to \NN_0\cup \{\infty\}$ which maps 
a vertex $v \in V(\TG)$ and time~$t\in \set{\lifetime}$ 
to the length of a 
shortest temporal $v$-$z$~path in~$\TG$ 
that departs at a time at least~$t$.
If such a temporal path does not exist, then~$d_\TG(v,t)=\infty$.
We drop the subscript~$\TG$ if it is clear from the context.

Intuitively, we now arrange all vertex appearances $(v,t)$ in the plane where the $x$-axis 
describes the distance (via temporal paths) of $v$ to $z$ at time $t$ and 
the $y$-axis describes the time.
Thus, $(v,t)$ gets the point $\left(d(v,t),t\right)$.
Consider \cref{fig:area-idea} for a moment.
We want to visualize a temporal $s$-$z$~path $P$ in this figure.
To this end, we say that $P$ visits vertex appearance~$(v,t)$ if~$P$ visits~$v$ in time step~$t$.
Hence, we can depict a temporal path $P$ by connecting the vertex appearances which are visited by $P$ in the visiting order.
Note that no temporal $v$-$z$ path or walk moves downwards. %
Moreover, among all temporal $v$-$z$~paths that depart at a time of at least~$t$, 
the shortest of them move with each time-edge further towards~$z$ (i.e., to the left).
For example, the dotted line in \cref{fig:area-idea} depicts the trajectory of 
a shortest temporal~$s$-$z$~path with a departure time~$t$.
The temporal path departs at time $t$ and arrives at time $\tau$.
This is not the case for a shortest $\wait$-restless temporal~$s$-$z$~path~$P$---such a temporal path can move to the right or stay at the same point while visiting multiple vertices.
For example, the solid (blue) line in \cref{fig:area-idea} depicts the trajectory of 
a shortest $\wait$-restless temporal~$s$-$z$~path.
Let~$\ell \defeq  k - d(s,1)$.
A crucial observation now is that if $P$ moves ``too far'' to the right or 
stays for ``too long'' at the same spot in the $x$-axis 
while visiting multiple vertices,
then $P$ would be too far away from $z$ (in terms of temporal paths distance) such that~$P$ cannot be of length at most~$k$.
This will lead us to the observation that 
for at least every $(2\ell+1)$-st vertex~$v$ which is visited by $P$ (at time $t$), 
the vertex appearances~$(v,t)$ has the following separation property:
\begin{enumerate}[(i)]
		\item each vertex appearance $(u,t')$ that $P$ visits before $v$ (hence, $v\not=u$)
				is to the right of $(v,t)$ and thus further away from $z$ than $(v,t)$, and
		\item each vertex appearance $(u,t')$ that $P$ visits after $v$ (hence, $v\not=u$) 
				is to the left of $(v,t)$ and thus closer to $z$ than $(v,t)$.
\end{enumerate}
\newcommand{\Area}[2]{\TG_{#1}^{#2}}
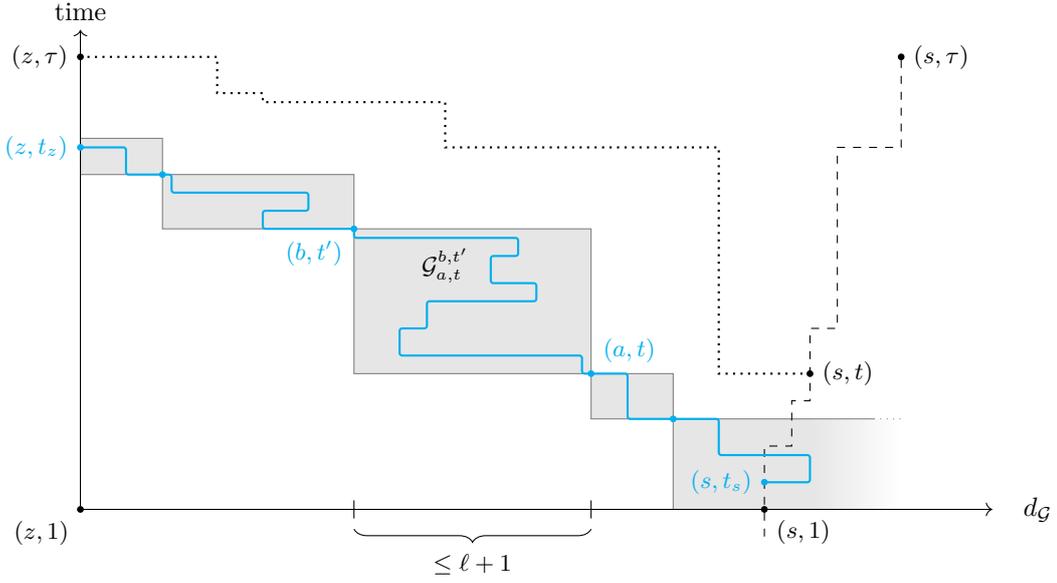
\begin{figure}[ht]
		\centering
		\begin{tikzpicture}[yscale=0.12,xscale=0.12]

			\fill [left color=gray!20!white, right color=white,draw=white] (80,10) rectangle (90,00);

			\fill [gray!20!white] (0,41) rectangle (9,37);
			\draw [gray] (0,41) rectangle (9,37);

			\fill [gray!20!white] (9,37) rectangle (30,31);
			\draw [gray] (9,37) rectangle (30,31);

			\fill [gray!20!white] (30,31) rectangle (56,15);
			\draw [gray] (30,31) rectangle (56,15);

			\fill [gray!20!white] (56,15) rectangle (65,10);
			\draw [gray] (56,15) rectangle (65,10);

			\fill [gray!20!white,draw=gray!20!white] (65,10) rectangle (80,00);
			\draw [gray] (65,0) -- (65,10) -- (87,10);

			\draw [dotted, gray] (87,10) -- (90,10);

			\draw[->] (0,0) -- (0,53);
			\draw[->] (0,0) -- (100,0);
			\draw[-,dashed] (75,-3) -- (75,7) -- (78,7) -- (78,12) -- (80,12) -- (80,20) -- (83,20) -- (83,40) -- (90,40) -- (90,50);
			\node at (105,0) {$d_\TG$};
			\node at (0,55) {time};
			\node[scale=0.5,vertex,label={230:{\small$(z,1)$}}] at (0,0) {};
			\node[scale=0.5,vertex,label={left:{\small$(z,\tau)$}}] at (0,50) {};
			\node[scale=0.5,vertex,label={330:{\small$(s,1)$}}] at (75,0) {};
			\node[scale=0.5,vertex,label={right:{\small$(s,\tau)$}}] at (90,50) {};
			
			\node[scale=0.5,vertex,label={right:{\small$(s,t)$}}] at (80,15) {};

			\node[scale=0.5,vertex,label={182:{\textcolor{cyan}{\small$(s,t_s)$}}},cyan] at (75,3) {};
			\node[scale=0.5,vertex,label={180:{\textcolor{cyan}{\small$(z,t_z)$}}},cyan] at (0,40) {};
			
			\path[draw,black,thick,dotted] (80,15) -- (75,15) -- (75,15) -- (70,15) -- (70,40) -- (40,40) 
					-- (40,45) -- (20,45) -- (20,46) -- (15,46) -- (15,50) -- (0,50) ;

			\node[scale=0.5,vertex,cyan] at (9,37) {};
			\node[scale=0.5,vertex,label={230:{\textcolor{cyan}{\small$(b,t')$}}},cyan] at (30,31) {};
			\node[scale=0.5,vertex,label={30:{\textcolor{cyan}{\small$(a,t)$}}},cyan] at (56,15) {};
			\node[scale=0.5,vertex,cyan] at (65,10) {};
			\node at (40,27) {\small $\Area{a,t}{b,t'}$};

			\draw [cyan,thick,rounded corners=1pt] 
					plot [ tension=2] coordinates { (75,3) (80,3) (80,6) (70,6) (70,10) (60,10) (60,15) (55,15) (55,17) (35,17) 
							(35,20) (38,20) (38,23) (50,23) (50,25) (45,25) (45,28) (48,28) (48,30) (30,30) (30,31) (20,31) 
			(20,33) (25,33) (25,35) (10,35)  (10,37) (5,37) (5,40) (0,40)};

			\draw[black] (30,1) -- (30,-1);
			\draw[black] (56,1) -- (56,-1);
			\draw [decorate,decoration={brace,amplitude=5pt,mirror,raise=4pt},yshift=0pt]
(30,-1) -- (56,-1) node [below,black,midway,yshift=-10pt] {\small $\leq \ell+1$};

		\end{tikzpicture}
		\caption{Illustration of the idea behind the dynamic programming table which is used to show \cref{thm:restless-abl}.
				The $y$-axis describes the time.
				The $x$-axis describes the distance to $z$ (via temporal paths).
				In this plane, a vertex appearance $(v,t)$ gets the position $\left(d(v,t),t\right)$.
				The positions of the vertex appearances of~$s$ are on the dashed line.
				A shortest (non-$\wait$-restless) temporal $s$-$z$~path that departs at time~$t$ is depicted by the dotted line.
				The trajectory of a shortest $\wait$-restless temporal $s$-$z$~path which departs at time~$t_s$
				and arrives at time~$t_z$ 
				is depicted by the solid (blue) line.
				Each gray area depicts a temporal subgraph which we use to compute $\wait$-restless paths
				from the vertex appearance on the right-bottom corner to the vertex appearance on
				the left-top corner, e.g., the temporal graph~$\Area{a,t}{b,t'}$.
		}
		\label{fig:area-idea}
\end{figure}
Moreover, we will observe that two consecutive vertex appearances which have this separation property, have a similar distance to $z$---the distances differ by at most $\ell+1$.
In \cref{fig:area-idea}, these special vertex appearances are at the left-top and right-bottom corners of each gray area.
Our dynamic program tries to guess these vertex appearances and then constructs for each gray area in \cref{fig:area-idea} 
a temporal graph that contains the $\wait$-restless path
from the right-bottom corner to the left-top corner of this area. 
Since we know that these $\wait$-restless temporal paths have length at most $2\ell+1$, 
we can use the algorithm developed in the last section to find them.

Another crucial observation we are going to make is that
two $\wait$-restless temporal paths from the right-bottom corner to
the left-top corner of two distinct gray areas in \cref{fig:area-idea} 
cannot visit the same vertex (except for their endpoints).
This is the case because the distance of a vertex~$v$ to~$z$ can only increase as time goes by.
Thus, if we find for each gray area in \cref{fig:area-idea} a $\wait$-restless temporal path from the right-bottom corner to the left-top corner, 
then this gives us a $\wait$-restless temporal $s$-$z$~path.
Henceforth the details follow.

Before we describe the dynamic programming table in \cref{sec:dp}, 
we define the temporal graph that contains all (shortest) $\wait$-restless paths in a gray area of \cref{fig:area-idea}.
To this end, we first define sets containing all vertex apperances of such a gray area.
For vertex appearances $(a,t), (b,t') \in V(\TG) \times \set{\lifetime}$,
we define
\begin{align*}
\mathcal A_{a,t}^{b,t'} \defeq & 
\left\{ 
(w,t^*) \in V(\TG) \times \set{\lifetime} \mvert  
 d(b,t') <  d(w,t^*) < d(a,t), t^* \in [t,t']
\right\} \text{ and}\\
\mathcal A_{}^{b,t'} \defeq & 
\left\{ 
(w,t^*) \in V(\TG) 
\times 
\set{\lifetime} \mvert  
\infty >
d(w,t^*) > d(b,t'), t^*\leq t'
\right\}.
\end{align*}
Now, the temporal graph~$\Area{a,t}{b,t'}$
for the gray area between~$(a,t)$ and~$(b,t')$ with~$t \leq t'$ is
defined by
\begin{align*}
\TE\left(\Area{a,t}{b,t'}\right) \defeq & 
\left\{ 
(\{v,u\},t^*) \in \TE(\TG) 
\mvert 
(v,t^*),(u,t^*) \in \mathcal A_{a,t}^{b,t'}
\right\}\\
&\cup
\left\{ 
(\{a,v\},t) \in \TE(\TG) 
\mvert  
(v,t) \in \mathcal A_{a,t}^{b,t'}
\right\}\\
&\cup
\left\{ 
(\{v,b\},t^*) \in \TE(\TG) 
\mvert  
 t' -\wait \leq t^*,
 (v,t^*) \in \mathcal A_{a,t}^{b,t'} \cup \{ (a,t) \}
\right\} \text{ and }\\
V\left(\Area{a,t}{b,t'}\right) \defeq & 
\left\{ 
v \in V(\TG) 
\mvert 
\exists (e,t^*) \in  \TE\left(\Area{a,t}{b,t'}\right) \colon v \in e
\right\}.
\end{align*}
For the gray area containing $s$ we have to adjust the definition of the corresponding temporal graph slightly.
To this end, we define~$\Area{}{b,t'}$ with
\begin{align*}
		\TE\left(\Area{}{b,t'}\right) \defeq & 
\left\{ 
(\{v,u\},t^*) \in \TE(\TG) 
\mvert 
(v,t^*),(u,t^*) \in \mathcal A_{}^{b,t'}
\right\}\\
\cup&
\left\{ 
(\{v,b\},t^*) \in \TE(\TG) 
\mvert  
 t' -\wait \leq t^*,
(v,t^*) \in \mathcal A_{}^{b,t'}
\right\} \text{ and }\\
V\left(\Area{}{b,t'}\right) \defeq & 
\left\{ 
v \in V(\TG) 
\mvert 
\exists (e,t^*) \in  \TE\left(\Area{}{b,t'}\right) \colon v \in e
\right\}.
\end{align*}
In the forthcoming section, we will use these definitions to solve \probSRestlessPath{}.

\subsection{The Dynamic Programming Table}
\label{sec:dp}
In this section, we describe the table~$T$ which we are going to use for the dynamic programming,
and show its correctness.

Intuitively, the table~$T$ has for each 
vertex appearance~$(u,t')$ an entry, and
if this entry contains a number~$p<\infty$,
then~$p$ is the length of the shortest $\wait$-restless 
temporal~$s$-$u$~path that only visits vertex appearances
which are, in \cref{fig:area-idea}, below and to the right 
of~$(u,t')$.

Let~$ I \defeq (\TG,s,z,\wait,k)$ be an instance of \probSRestlessPath{},
where $k = d(s,1)+\ell$.
For all~$(u,t') \in V(\TG) \times \set{\lifetime}$
such that there is an~$e \in E_{t'}$ with~$v \in e$,
we define~$T$ as follows.
If $d(s,1) - d(u,t') \leq \ell$,
then
\begin{align}
		\label{dp:alb1}
T[u,t'] \defeq  
\begin{cases}
		0, & \text{if $u = s$;}\\
\ell', & \text{if $u \not= s$ and $\ell' \in [2\ell]$ is the length of a} \\
	   & \text{shortest $\wait$-restless $s$-$u$~path in $\Area{}{u,t'}$;} \\
\infty, & \text{otherwise.}
\end{cases}
\end{align}
If $d(s,1) - d(u,t') > \ell$,
then
\begin{align}
\label{dp:alb2}
T[u,t'] \defeq  \min \left( \{\infty\} \cup \left\{ T[v,t] + \ell' \middle\vert 
\begin{array}{r}
		t \in \set{t'}, e \in E_t, v \in e \text{, where}\\
d(v,t) > d(u,t') \geq d(v,t) - \ell-1\\
\text{and $\ell' \in [2\ell+1]$ is the length of a } \\
\text{shortest $\wait$-restless $v$-$u$~path in $\Area{v,t}{u,t'}$}
\end{array}
\right\}
\right)
\end{align}
In the end, we will report that $I$ is a \yes-instance if and only if
there is a $t \in \set{\lifetime}$ such that $T[z,t] \leq k$.
We will show the correctness of this in the following lemmata.
We start with the backwards direction.
\begin{lemma}
		\label{lem:restless-alb-forward}
		Let $(\TG,s,z,\wait,k)$ be an instance of \probSRestlessPath{}.
		If $T[z,t_z]\leq k < \infty$ (defined in \eqref{dp:alb1} and \eqref{dp:alb2}), 
		then there is a $\wait$-restless temporal $s$-$z$ path of length at most $k$ in $\TG$.
\end{lemma}
\begin{proof}
		We show by induction on the distance to $z$  that
		if $T[u,t'] = k' < \infty$, then there is a $\wait$-restless $s$-$u$~path of length $k'$ in~$\Area{}{u,t'}$
		which arrives at $u$ at some time step in $[t'-\wait,t']$.
		
		Note that all  temporal $s$-$u$~paths in $\Area{}{u,t'}$ arrive at some time in $[t'-\wait,t']$.
		By~\eqref{dp:alb1}, for each vertex appearance~$(u,t')$ with $d(s,1)-d(u,t') \leq \ell$ the induction hypothesis is true---this is our base case.

		Now let $(u,t')$ be a vertex appearance with $T[u,t'] = k' < \infty$.
		Assume that for all vertex appearances $(v,t)$ with $d(v,t) > d(u,t')$ we have that
		if $T[v,t] = k'' < \infty$, then there is a $\wait$-restless temporal $s$-$v$~path of length $k''$ in~$\Area{}{v,t}$
		which arrives at~$v$ at some time step in $[t-\wait,t]$.
		Since $T[u,t'] = k'$, we know by \eqref{dp:alb2} that 
		there is a vertex appearance $(v,t)$ with~$T[v,t]=k''$, $t\leq t'$, and~$d(v,t) > d(u,t')$.
		Moreover, there is a $\wait$-restless temporal $v$-$u$~path~$P_2$ in~$\Area{v,t}{u,t'}$ of length~$\ell'= k'-k''$.
		By the definition of~$\Area{v,t}{u,t'}$, $P$ departs at time~$t$ and arrives at some time in~$[t'-\wait,t']$.
		By assumption, there is a $\wait$-restless temporal $s$-$v$~path~$P_1$ of length~$k''$ in~$\Area{}{v,t}$
		which arrives at~$v$ at some time step in~$[t-\wait,t]$.
		We now append the time-edges of $P_2$ to the time-edges of $P_1$ and claim that 
		the resulting time-edge sequence $P$ is a $\wait$-restless temporal $s$-$u$~path of length~$k'$ which arrives at~$u$
		at some time in~$[t'-\wait,t']$.
		Observe that $P$ is a $\wait$-restless temporal $s$-$u$~walk of length $k' = k''+\ell'$,
		as
		\begin{itemize}
			\item $P_1$ is $\wait$-restless, of length~$k''$, and arrives at $v$ at some time $t^* \in [t-\wait,t]$, and
			\item $P_2$ is of length $\ell'$ and departs at time~$t$.
		\end{itemize}
		Moreover, the arrival time of $P$ is the same as the arrival time of $P_2$.

		It remains to show that $P$ does not visit a vertex twice.
		To see this, we show that~$V(\Area{}{v,t}) \cap V(\Area{v,t}{u,t'}) = \{v\}$.
		This will complete the proof, since
		we know that~$V(P_1) \subseteq V(\Area{}{v,t})$, $V(P_2) \subseteq V(\Area{v,t}{u,t'})$, $P_1$ ends at vertex~$v$, and~$P_2$ starts at vertex~$v$.
		By definition, we have that~$v \in (V(\Area{}{v,t}) \cap V(\Area{v,t}{u,t'}))$.
		Assume towards a contradiction that there is a vertex~$w \in  (V(\Area{}{v,t}) \cap V(\Area{v,t}{u,t'}))\setminus \{v\}$.
		Then, there must be time steps $t_1,t_2$ such that $(w,t_1) \in \mathcal A^{v,t} \cup \{ (u,t') \}$ 
		and $(w,t_2) \in \mathcal A_{v,t}^{u,t'}$.
		Note that $d(w,t_1) > d(v,t) > d(w,t_2)$ and 
		hence each temporal $w$-$z$~path in $\TG$ that 
		departs not earlier than $t_1$ is 
		longer than a shortest $w$-$z$~path in $\TG$ that 
		departs not earlier than $t_2$.
		This is a contradiction because $t_1 \leq t \leq t_2$.
\end{proof}

To show the forward direction of the correctness, 
we introduce further notation.
Recall from the definition of the dynamic programming table~$T$ in \eqref{dp:alb1} and \eqref{dp:alb2} that $\ell = k - d(s,1)$.
Assume the input instance $I$ is a \yes-instance.
Thus there is a $\wait$-restless temporal $s$-$\z$~path $P = ((\{v_{i-1},v_i\},t_i))_{i=1}^k$
of length at most $d(s,1)+\ell$ in $\TG$. 
Let $s=v_0,v_1,\dots,v_{k}=z$ be the order in which $P$ visits the vertices in $V(P)$.
For simplicity, let $t_0 \defeq 1$ and $t_{k+1} \defeq  t_k$.
For all $i \in [0,k]$,
	we say that $v_i$ is a \emph{distance separator} if 
	\begin{enumerate}[(i)]
			\item $d(v_i,t_{i+1}) < d(v_j,t_{j+1})$ for all $j \in [0,i-1]$, and
			\item $d(v_i,t_{i+1}) > d(v_j,t_{j+1})$ for all $j \in [i+1,k]$.
	\end{enumerate}

	Before we show the forward direction of the correctness of the dynamic programming table~$T$,
	we show that $P$ visits a distance separator on regular basis.
	\begin{lemma}
			\label{a-claim}
	For all~$i \in [0,k]$ there is a $j\in [0,2\ell]$ such that $v_{i+j}$ is a distance separator.
	\end{lemma}
	\begin{proof}
		We show this statement with a reverse induction on the length of $P$.
	As $z$ is clearly a distance separator, the claim is true for all values in $[\max\{0,k-2\ell\},k]$. 
	This is the base case of our induction.

	Let $k - 2\ell > 0$ and let $i \in [0,k-2\ell-1]$ and
	assume that for all $i' \in [i+1,k]$ there is a  $j'\in [0,2\ell]$ such that~$v_{i'+j'}$ is a distance separator.
	Let $n \in [i',k]$ be the smallest possible number such that $v_n$ is a distance separator.
	Let $f \in [k]$  be the smallest possible number such that $d(v_f,t_{f+1}) - d(v_n,t_{n+1}) = \ell+1$.
	Note that if such an $f$ does not exist, then the claim is true.
	Hence, we assume that such an $f$ exists.
	Note that $f \leq i$, otherwise $n$ is not the smallest possible number.

	We now show that $n - f \leq 2\ell+1$.
	Assume towards a contradiction
	that the temporal $v_f$-$v_n$~path contained in $P$ has length $n - f > 2\ell+1$.
	This is a lower bound for the length of $P$.
	We get the following.
	\begin{align*}
			 d(s,1) - d(v_f,t_{f+1}) + 2\ell+1 + d(v_n,t_{n+1}) &< k = d(s,1) + \ell\\
			\implies d(v_n,t_{n+1}) - d(v_f,t_{f+1}) + \ell+1  &< 0\\
			\implies  \ell+1  &< d(v_f,t_{f+1}) + d(v_n,t_{n+1}) 
	\end{align*}
	This is a contradiction to $d(v_f,t_{f+1}) - d(v_n,t_{n+1}) = \ell+1$.

	Next, we show that, between $v_{f}$ and $v_{n-1}$, $P$ must visit a distance separator. 
	Assume towards a contradiction that $v_{n-j''}$ is not a distance separator, for all~$j'' \in [n-f]$.
	Hence, for all $p \in [d(v_n,t_{n+1})+1, d(v_f,t_{f+1})]$, 
	there are two distinct $q,r \in [f,n-1]$ such that
	$d(v_r,t_{r+1}) = d(v_q,t_{q+1}) = p$.
	Since~$[d(v_n,t_{n+1})+1, d(v_f,t_{f+1})]=\ell$, we get by the pigeonhole principle
	that the temporal $v_f$-$v_n$~path contained in $P$ has length $n - f > 2\ell+1$---a contradiction.
\end{proof}

Finally, we are set to show the forward direction.
\begin{lemma}
		\label{lem:restless-abl-backwards}
		Let $(\TG,s,z,\wait,k)$ be an instance of \probSRestlessPath{}.
		If there is a $\wait$-restless temporal $s$-$z$ path in $\TG$ of length at most $k$ with arrival time $t_k$,
		then $T[z,t_k]\leq k$ (defined in \eqref{dp:alb1} and \eqref{dp:alb2}).
\end{lemma}
\begin{proof}
	
Let $P = ((\{v_{i-1},v_i\},t_i))_{i=1}^k$ be a shortest $\wait$-restless temporal $s$-$\z$~path
of length at most $k = d(s,1)+\ell$ in $\TG$. %
Let $s=v_0,v_1,\dots,v_{k}=z$ be the order in which $P$ visits the vertices in $V(P)$.
For simplicity, let $t_0 \defeq 1$ and $t_{k+1} \defeq  t_k$.
Moreover, let $m$ be the number of distance separators visited by $P$ and
let $\sigma \colon \set{m} \to [0,k]$ be an injective function such that $v_{\sigma(i)}$ is the $i$-th distance separator
which is visited by $P$ (from $s$ to $z$), for all $i \in \set{m}$.
Note that, the vertex~$v_{\sigma(i)}$ is the $i$-th distance separator visited by $P$ and thus $\sigma(i)$ also 
describes the length of the $\wait$-restless temporal $s$-$v_{\sigma(i)}$~subpath contained in $P$.

We now show that for all $i \in \set{m}$ we have that $T[v_{\sigma(i)},t_{\sigma(i)+1}] \leq \sigma(i)$.
	If $\sigma(1)=0$, then $s$ is a distance separator and the claim is clearly true, see \eqref{dp:alb1}.
	Otherwise, by \cref{a-claim}, we have %
	$\sigma(1) \leq 2\ell$.
	Hence, $P$ contains a $\wait$-restless 
	temporal $s$-$v_{\sigma(1)}$~path of length $\sigma(1) \leq 2\ell$ which is contained in $\Area{}{v_{\sigma(1)},t_{\sigma(1)+1}}$.
	Thus, $T[v_{\sigma(1)},t_{\sigma(1)+1}] \leq \sigma(1)$.

	Now assume that for some $i \in [2,m]$ we have that $T[v_{\sigma(i-1)},t_{\sigma(i-1)+1}] \leq \sigma(i-1)$.
	Observe that $t_{\sigma(i-1)+1} \leq  t_{\sigma(i)+1}$ and that $d(v_{\sigma(i-1)},t_{\sigma(i-1)+1}) > d(v_{\sigma(i)},t_{\sigma(i)+1})$.
	By \cref{a-claim}, we have that $\sigma(i) - \sigma(i-1) \leq 2\ell +1$ and
	that the $\wait$-restless temporal $v_{\sigma(i-1)}$-$v_{\sigma(i)}$~path $Q$ contained in $P$ is of length $\sigma(i) - \sigma(i-1) \leq 2\ell +1$.
	As all of the at most $2\ell$ vertices in $V(Q) \setminus \{ v_{\sigma(i-1)}, v_{\sigma(i)}\}$
	are not distance separators, we have by the pigeonhole principle that $d(v_{\sigma(i-1)},t_{\sigma(i-1)+1}) - d(v_{\sigma(i)},t_{\sigma(i)+1}) \leq \ell+1$.
	Moreover, note that $Q$ 
	in $\Area{v_{\sigma(i-1)},t_{\sigma(i-1)+1}}{v_{\sigma(i)},t_{\sigma(i)+1}}$, because $v_{\sigma(i-1)}$ and $v_{\sigma(i)}$ are distance separators.
	Hence, by \eqref{dp:alb2}, we have that $T[v_{\sigma(i)},t_{\sigma(i)+1}] \leq T[v_{\sigma(i-1)},t_{\sigma(i-1)+1}] + \sigma(i) - \sigma(i-1) \leq \sigma(i)$,
	as we have $T[v_{\sigma(i-1)},t_{\sigma(i-1)+1}] \leq \sigma(i-1)$ by assumption.

	Since $z=v_k$, we have that $k$ is the only number in $[0,k]$ with $d(v_k,t_{k+1})=0$.
	Hence, $v_k$ is the last distance separator and thus $\sigma(m) = k$.
	Finally, by our induction, we have that $T[z,t_{k+1}] \leq k$.
\end{proof}
\subsection{Putting the Pieces Together}
\label{sec:final-poof}
In this section, we finally show \cref{thm:restless-abl}.
Towards this end, we first show that we can compute all necessary values of our distances function $d(\cdot,\cdot)$ in linear time.
\begin{lemma}
		\label{lem:compute-distances}
	Given a temporal graph $\TGcompact$ and a vertex $z$,
	one can compute in~$O(\abs{\TG})$~time the value $d(v,t)$, for all $v\in V$ and $t \in [\lifetime]$ where~$v$ is not 
	isolated in the graph~$(V,E_t)$.
\end{lemma}
\begin{proof}
		We will construct a directed graph~$D$ where each arc have either weight zero or one
		such that the weight of a shortest $z$-$v_t$~path equals the value of~$d(v,t)$, 
		for all $v\in V$ and $t \in [\lifetime]$ where~$v$ is not isolated in the graph~$(V,E_t)$.
		Then, a slightly modified breadth-first search will do the task.

		We compute the set $\mathcal V$ of non-isolated vertex appearances. 
		That is, $\mathcal V \defeq  \{ (v,t) \in V(\TG) \times \set{\lifetime} \mid \exists e \in E_t\colon v \in e\}$.
		Note that this can be done in $O(\abs{\TG})$ time and that $\abs{\mathcal V} \leq 2\abs{\TG}$.
		Now we are ready to define $D$ by
		\begin{align*}
				V(D) \defeq\ &\{ z \} \cup  \{ v_t \mid (v,t) \in \mathcal V\}\\
				E(D) \defeq\ &\{ (v_t,u_t),(u_t,v_t) \mid (v,t),(u,t) \in \mathcal V \text{ and } u \neq v \}\ \cup\\
					 &\left\{ (v_{t_2},v_{t_1}) \mvert v_{t_1},v_{t_2} \in V(D) \text{ and } 
				t_2 = \min \left\{ t \mvert (v,t) \in \mathcal V \text{ and } t > t_1 \right\} \right\}\ \cup\\
					 &\{ (z,z_t) \mid z_t \in V(D) \text{ and } t = \max \{ t' \mid (z,t') \in \mathcal V \} \}.
		\end{align*}
		Now all arcs in $\{ (v_t,u_t),(u_t,v_t) \mid (v,t),(u,t) \in \mathcal V \text{ and } u \neq v \}$ get weight one, 
		while all the other arcs get weight zero.
		Note that the $V(D)+E(D) \in O(\abs{\TG})$ and that $D$ can be constructed in~$O(\abs{\TG})$~time.
		Observe that for every temporal $v$-$z$~path $P$ in $\TG$ with departure time~$t$
		there is a $z$-$v_t$~path in $D$ whose accumulated edge-weight equals the length of~$P$.
		Hence, if we know the minimum edge-weight of the paths from~$z$ to all vertices in $D$,
		then we also know the value $d(v,t)$, 
		for all $v\in V$ and $t \in [\lifetime]$ where~$v$ is not isolated in the graph~$(V,E_t)$.
		Thus, we employ a breadth-first search that starts at~$z$ and only explores an arc of weight one
		of there is currently no arc of weight zero which could be explored instead.
		At each vertex $v_t \in V(D)$ we store the edge-weight $d(v,t)$ of the path from $z$ to this vertex.
		Hence, the overall running time of this procedure is $O(\abs{\TG})$~time.
\end{proof}

Finally, we are set to show \cref{thm:restless-abl}:
\mainthm{}

\begin{proof}[Proof of \cref{thm:restless-abl}]
		Let $I\defeq (\TG,s,z,\wait,k)$ be an instance of \probSRestlessPath{}.
		We perform the following.
		First, we compute the set $\mathcal V$ of non-isolated vertex appearances. 
		That is, $\mathcal V \defeq  \{ (v,t) \in V(\TG) \times \set{\lifetime} \mid \exists e \in E_t\colon v \in e\}$.
		Note that this can be done in $O(\abs{\TG})$ time and that $\abs{\mathcal V} \leq 2\abs{\TG}$.
		By \cref{lem:compute-distances}, we compute~$d(v,t)$ for all~$(v,t) \in \mathcal V$ in~$O(\abs{\TG})$~time.
		We may assume that there is a temporal $s$-$z$~path in $\TG$ and that a shortest of them has length at most~$k$, 
		otherwise $I$ is clearly a \no-instance.
		We set~$\ell \defeq  k - d(s,1) = k - d(s,t)$, 
		where $t=\min \{t' \in \set{\lifetime} \mid (s,t') \in \mathcal V \}$.
		Note that table~$T$, defined in \eqref{dp:alb1} and 
		\eqref{dp:alb2}, 
		has $O(\abs{\TG})$~entries---one 
		for each element in $\mathcal V$.
		To compute one entry in $T$, we consider at most $O(\abs{\TG})$ other entries in $T$ and 
		for each of them we have to check at most~$2\ell+1$ times 
		whether a temporal graph of size $O(\abs{\TG})$ 
		has a $\wait$-restless temporal path of length~$\ell' \in [2\ell+1]$ between two distinct vertices.
		We answer each of these checks by \cref{lem:exact-restless-path} 
		with a one-sided error probability of at most~$p'$
		in~$O(4^{\ell} \cdot \ell \abs{\TG}\wait\log(\ell\cdot \nicefrac{1}{p'}))$ time.
		How we set the error probability~$p'$ will be determinate in a moment.	

		We say that $I$ is a \yes-instance if and only if
		there is a $(z,t) \in \mathcal V$ such that~$T[z,t] \leq k$.
		If \cref{lem:exact-restless-path} reports \yes, then with probability one, 
		there is such a $\wait$-restless temporal path in question.
		Hence, by \cref{lem:restless-alb-forward}, if our overall algorithm reports \yes, then $I$ is a \yes-instance---the error probability is zero in this case.
		If our overall algorithm reports \no, 
		then the probability that $I$ is a \yes-instance shall be at most $1-p$.
		By \cref{lem:restless-abl-backwards}, it remains to determine~$p'$.
		Recall from \cref{a-claim} that a $\wait$-restless temporal $s$-$\z$~path~$P$ of length at most~$k$ visits at least every $2\ell+1$~vertices one distance separator.
		Hence, we can identify at most $\lceil\nicefrac{k}{\lceil\ell + \nicefrac{1}{2}\rceil}\rceil \in O(\nicefrac{k}{\ell})$ vertex appearances which are visited by $P$
		and thus $O(\nicefrac{k}{\ell})$ calls of the algorithm behind \cref{lem:exact-restless-path}
		such that if these calls are answered correctly then this causes our overall algorithm to report \yes, as $T[z,t] \leq k$ for some $t \in \set{\lifetime}$.
		Hence, there is a $p' \in O(\nicefrac{p\ell}{k})$ such that
		we have an error probability of at most~$p$ in the case our overall algorithm answers \no.
		Thus, we can compute all entries of~$T$ in~$\runningtimeInL$
				time, where~$\ell := k - d$ and $d$~is the minimum length of a temporal $s$-$\z$~path.
\end{proof}
On a more practical note,
one can observe that in order to compute one entry for vertex appearances~$(u,t)$ of table~$T$,
we only consider table entries of vertex appearances which are close to $(u,t)$ in terms of the distance~$d(\cdot,\cdot)$.
Thus, for temporal graphs that are nowhere dense in terms of $d(\cdot,\cdot)$, 
it seems reasonable that the presented dynamic programming technique does not induce a quadratic running time, 
in terms of the temporal graph size, 
on top of running time of \cref{lem:exact-restless-path}.
For example in contact networks where mass events are prohibited. 
Moreover, a $\wait$-restless temporal path of length $k$ has a time horizon of at most $(k-1)\wait$.
Hence, with an overhead of $O(\lifetime)$ one could guess the departure time $t$ of the $\wait$-restless $s$-$\z$~path and
discard all time-edge $(e,t')$ with $t' < t$ or $t' > t+(k-1)\wait+1$.
This potentially decreases the parameter $k-d$ and thus the exponential part of the running time substantial,
where $d$ is the minimum length of a temporal $s$-$\z$~path.

\section{Conclusion}

We showed that \probSRestlessPath{} admits fixed-parameter tractability for parameters below the solution size $k$.
In particular, we showed that \probSRestlessPath{} can be solved in $4^{k-d}\cdot |\TG|^{O(1)}$ time
with a one-sided error probability of at most~$2^{-\abs{\TG}}$, 
where~$d$ is the minimum length of a temporal $s$-$z$~path.
In the corresponding algorithm, we have only one subroutine with a super-polynomial running time:
an algorithm to find a $\wait$-restless temporal path of length at most $2(k-d)+1$.
Moreover, this is also the only subroutine that has a non-zero error probability. 

We believe that our algorithmic approach opens new research directions
to advance further:
\begin{itemize}
		\item First, we wonder how good our algorithm performs in an experimental setup comparable to the one of 
		 Thejaswi~\etal~\cite{thejaswi2020restless}.
 \item Second, one could study in detail the temporal subgraphs on which we employ \cref{lem:exact-restless-path}.
		 In principle, \cref{lem:exact-restless-path}, could be replaced with any other algorithm for \probSRestlessPath{}.
		 Do these specific temporal subgraphs admit structural properties which are algorithmically useful?
 \item Third, we believe that our geometric perspective presented in \cref{sec:general-idea} can be applied to other temporal graph problems.
		 In particular, for temporal graph problems which ask for specific temporal paths, e.g., 
		 temporal paths that obey certain robustness properties \cite{FMNR22}, or
		 temporal paths that visit all vertices at least once \cite{Erlebach0K15,ErlebachKLSS19,michail2016traveling} parameterized 
		 by the temporal diameter, that is, the length of the longest shortest temporal path between two arbitrary vertices.
\end{itemize}

\bibliography{strings-long,references}

\begin{thebibliography}{10}

\bibitem{Akr+19a}
Eleni~C. Akrida, Jurek Czyzowicz, Leszek G{\k{a}}sieniec, {\L{}}ukasz Kuszner,
  and Paul~G. Spirakis.
\newblock Temporal flows in temporal networks.
\newblock {\em Journal of Computer and System Sciences}, 103:46--60, 2019.
\newblock \href {https://doi.org/10.1016/j.jcss.2019.02.003}
  {\path{doi:10.1016/j.jcss.2019.02.003}}.

\bibitem{Akr+17}
Eleni~C. Akrida, Leszek G{\k{a}}sieniec, George~B. Mertzios, and Paul~G.
  Spirakis.
\newblock The complexity of optimal design of temporally connected graphs.
\newblock {\em Theory of Computing Systems}, 61(3):907--944, 2017.
\newblock \href {https://doi.org/10.1007/s00224-017-9757-x}
  {\path{doi:10.1007/s00224-017-9757-x}}.

\bibitem{AMS19}
Eleni~C. Akrida, George~B. Mertzios, Paul~G. Spirakis, and Christoforos~L.
  Raptopoulos.
\newblock The temporal explorer who returns to the base.
\newblock {\em Journal of Computer and System Sciences}, 120:179--193, 2021.
\newblock \href {https://doi.org/10.1016/j.jcss.2021.04.001}
  {\path{doi:10.1016/j.jcss.2021.04.001}}.

\bibitem{DBLP:journals/algorithmica/AlonGKSY11}
Noga Alon, Gregory~Z. Gutin, Eun~Jung Kim, Stefan Szeider, and Anders Yeo.
\newblock Solving {MAX-\emph{r}-SAT} above a tight lower bound.
\newblock {\em Algorithmica}, 61(3):638--655, 2011.
\newblock \href {https://doi.org/10.1007/s00453-010-9428-7}
  {\path{doi:10.1007/s00453-010-9428-7}}.

\bibitem{AxiotisF16}
Kyriakos Axiotis and Dimitris Fotakis.
\newblock On the size and the approximability of minimum temporally connected
  subgraphs.
\newblock In {\em Proceedings of the 43rd International Colloquium on Automata,
  Languages, and Programming (ICALP)}, volume~55 of {\em Leibniz International
  Proceedings in Informatics}, pages 149:1--149:14. Schloss Dagstuhl --
  Leibniz-Zentrum f{\"u}r Informatik, 2016.
\newblock \href {https://doi.org/10.4230/LIPIcs.ICALP.2016.149}
  {\path{doi:10.4230/LIPIcs.ICALP.2016.149}}.

\bibitem{Bar16}
Albert-L{\'a}szl{\'o} Barab{\'a}si.
\newblock {\em Network Science}.
\newblock Cambridge University Press, 2016.

\bibitem{himmel_efficient_2019}
Matthias Bentert, Anne-Sophie Himmel, Andr{\'{e}} Nichterlein, and Rolf
  Niedermeier.
\newblock Efficient computation of optimal temporal walks under waiting-time
  constraints.
\newblock {\em Applied Network Science}, 5(72):1--26, 2020.
\newblock \href {https://doi.org/10.1007/s41109-020-00311-0}
  {\path{doi:10.1007/s41109-020-00311-0}}.

\bibitem{DBLP:journals/siamdm/BezakovaCDF19}
Ivona Bez{\'{a}}kov{\'{a}}, Radu Curticapean, Holger Dell, and Fedor~V. Fomin.
\newblock Finding detours is fixed-parameter tractable.
\newblock {\em SIAM Journal on Discrete Mathematics}, 33(4):2326--2345, 2019.
\newblock \href {https://doi.org/10.1137/17M1148566}
  {\path{doi:10.1137/17M1148566}}.

\bibitem{BF03}
Sandeep Bhadra and Afonso Ferreira.
\newblock Computing multicast trees in dynamic networks and the complexity of
  connected components in evolving graphs.
\newblock {\em Journal of Internet Services and Applications}, 3(3):269--275,
  2012.
\newblock \href {https://doi.org/10.1007/s13174-012-0073-z}
  {\path{doi:10.1007/s13174-012-0073-z}}.

\bibitem{BodlaenderZ19}
Hans~L. Bodlaender and Tom~C. van~der Zanden.
\newblock On exploring always-connected temporal graphs of small pathwidth.
\newblock {\em Information Processing Letters}, 142:68--71, 2019.
\newblock \href {https://doi.org/10.1016/j.ipl.2018.10.016}
  {\path{doi:10.1016/j.ipl.2018.10.016}}.

\bibitem{XuanFJ03}
Binh{-}Minh Bui{-}Xuan, Afonso Ferreira, and Aubin Jarry.
\newblock Computing shortest, fastest, and foremost journeys in dynamic
  networks.
\newblock {\em International Journal of Foundations of Computer Science},
  14(2):267--285, 2003.
\newblock \href {https://doi.org/10.1142/S0129054103001728}
  {\path{doi:10.1142/S0129054103001728}}.

\bibitem{BM21}
Benjamin~M. Bumpus and Kitty Meeks.
\newblock Edge exploration of temporal graphs.
\newblock In {\em Proceedings of the 32st International Workshop on
  Combinatorial Algorithms (IWOCA)}, volume 12757 of {\em Lecture Notes in
  Computer Science}, pages 107--121. Springer, 2021.
\newblock \href {https://doi.org/10.1007/978-3-030-79987-8\_8}
  {\path{doi:10.1007/978-3-030-79987-8\_8}}.

\bibitem{CasteigtsHMZ20}
Arnaud Casteigts, Anne{-}Sophie Himmel, Hendrik Molter, and Philipp Zschoche.
\newblock Finding temporal paths under waiting time constraints.
\newblock {\em Algorithmica}, 83(9):2754--2802, 2021.
\newblock \href {https://doi.org/10.1007/s00453-021-00831-w}
  {\path{doi:10.1007/s00453-021-00831-w}}.

\bibitem{DBLP:journals/jcss/CrowstonFGJKRRTY14}
Robert Crowston, Michael~R. Fellows, Gregory~Z. Gutin, Mark Jones, Eun~Jung
  Kim, Fran Rosamond, Imre~Z. Ruzsa, St{\'{e}}phan Thomass{\'{e}}, and Anders
  Yeo.
\newblock Satisfying more than half of a system of linear equations over
  {GF(2):} {A} multivariate approach.
\newblock {\em Theory of Computing Systems}, 80(4):687--696, 2014.
\newblock \href {https://doi.org/10.1016/j.jcss.2013.10.002}
  {\path{doi:10.1016/j.jcss.2013.10.002}}.

\bibitem{cygan2015parameterized}
Marek Cygan, Fedor~V. Fomin, Lukasz Kowalik, Daniel Lokshtanov, D{\'{a}}niel
  Marx, Marcin Pilipczuk, Micha{\l} Pilipczuk, and Saket Saurabh.
\newblock {\em Parameterized Algorithms}.
\newblock Springer, 2015.
\newblock \href {https://doi.org/10.1007/978-3-319-21275-3}
  {\path{doi:10.1007/978-3-319-21275-3}}.

\bibitem{DeligkasP20}
Argyrios Deligkas and Igor Potapov.
\newblock Optimizing reachability sets in temporal graphs by delaying.
\newblock In {\em Proceedings of the 34th Conference on Artificial Intelligence
  (AAAI)}, pages 9810--9817, 2020.
\newblock \href {https://doi.org/10.1609/aaai.v34i06.6533}
  {\path{doi:10.1609/aaai.v34i06.6533}}.

\bibitem{Die16}
Reinhard Diestel.
\newblock {\em Graph Theory}, volume 173.
\newblock Springer, 5 edition, 2016.
\newblock \href {https://doi.org/10.1007/978-3-662-53622-3}
  {\path{doi:10.1007/978-3-662-53622-3}}.

\bibitem{DF13}
Rodney~G. Downey and Michael~R. Fellows.
\newblock {\em Fundamentals of Parameterized Complexity}.
\newblock Springer, 2013.
\newblock \href {https://doi.org/10.1007/978-1-4471-5559-1}
  {\path{doi:10.1007/978-1-4471-5559-1}}.

\bibitem{EnrightMMZ19}
Jessica Enright, Kitty Meeks, George~B. Mertzios, and Viktor Zamaraev.
\newblock Deleting edges to restrict the size of an epidemic in temporal
  networks.
\newblock {\em Journal of Computer and System Sciences}, 119:60--77, 2021.
\newblock \href {https://doi.org/10.1016/j.jcss.2021.01.007}
  {\path{doi:10.1016/j.jcss.2021.01.007}}.

\bibitem{enright2021assigning}
Jessica Enright, Kitty Meeks, and Fiona Skerman.
\newblock Assigning times to minimise reachability in temporal graphs.
\newblock {\em Journal of Computer and System Sciences}, 115:169--186, 2021.
\newblock \href {https://doi.org/10.1016/j.jcss.2020.08.001}
  {\path{doi:10.1016/j.jcss.2020.08.001}}.

\bibitem{Erlebach0K15}
Thomas Erlebach, Michael Hoffmann, and Frank Kammer.
\newblock On temporal graph exploration.
\newblock {\em Journal of Computer and System Sciences}, 119:1--18, 2021.
\newblock \href {https://doi.org/10.1016/j.jcss.2021.01.005}
  {\path{doi:10.1016/j.jcss.2021.01.005}}.

\bibitem{ErlebachKLSS19}
Thomas Erlebach, Frank Kammer, Kelin Luo, Andrej Sajenko, and Jakob~T. Spooner.
\newblock Two moves per time step make a difference.
\newblock In {\em Proceedings of the 46th International Colloquium on Automata,
  Languages, and Programming (ICALP)}, volume 132 of {\em Leibniz International
  Proceedings in Informatics}, pages 141:1--141:14. Schloss Dagstuhl --
  Leibniz-Zentrum f{\"u}r Informatik, 2019.
\newblock \href {https://doi.org/10.4230/LIPIcs.ICALP.2019.141}
  {\path{doi:10.4230/LIPIcs.ICALP.2019.141}}.

\bibitem{ErlebachS18}
Thomas Erlebach and Jakob~T. Spooner.
\newblock Faster exploration of degree-bounded temporal graphs.
\newblock In {\em Proceedings of the 43rd International Symposium on
  Mathematical Foundations of Computer Science (MFCS)}, volume 117 of {\em
  Leibniz International Proceedings in Informatics}, pages 36:1--36:13. Schloss
  Dagstuhl -- Leibniz-Zentrum f{\"u}r Informatik, 2018.
\newblock \href {https://doi.org/10.4230/LIPIcs.MFCS.2018.36}
  {\path{doi:10.4230/LIPIcs.MFCS.2018.36}}.

\bibitem{DBLP:series/txtcs/FlumG06}
J{\"{o}}rg Flum and Martin Grohe.
\newblock {\em Parameterized Complexity Theory}.
\newblock Texts in Theoretical Computer Science. Springer, 2006.
\newblock \href {https://doi.org/10.1007/3-540-29953-X}
  {\path{doi:10.1007/3-540-29953-X}}.

\bibitem{FluschnikMNRZ20}
Till Fluschnik, Hendrik Molter, Rolf Niedermeier, Malte Renken, and Philipp
  Zschoche.
\newblock Temporal graph classes: {A} view through temporal separators.
\newblock {\em Theoretical Computer Science}, 806:197--218, 2020.
\newblock \href {https://doi.org/10.1016/j.tcs.2019.03.031}
  {\path{doi:10.1016/j.tcs.2019.03.031}}.

\bibitem{FMNR22}
Eugen Füchsle, Hendrik Molter, Rolf Niedermeier, and Malte Renken.
\newblock Delay-robust routes in temporal graphs.
\newblock In {\em Proceedings of the 39th International Symposium on
  Theoretical Aspects of Computer Science (STACS)}, Leibniz International
  Proceedings in Informatics. Schloss Dagstuhl -- Leibniz-Zentrum f{\"u}r
  Informatik, 2022.
\newblock To appear.

\bibitem{DBLP:journals/mst/GutinKLM11}
Gregory~Z. Gutin, Eun~Jung Kim, Michael Lampis, and Valia Mitsou.
\newblock Vertex cover problem parameterized above and below tight bounds.
\newblock {\em Theory of Computing Systems}, 48(2):402--410, 2011.
\newblock \href {https://doi.org/10.1007/s00224-010-9262-y}
  {\path{doi:10.1007/s00224-010-9262-y}}.

\bibitem{DBLP:journals/jcss/GutinIMY12}
Gregory~Z. Gutin, Leo van Iersel, Matthias Mnich, and Anders Yeo.
\newblock Every ternary permutation constraint satisfaction problem
  parameterized above average has a kernel with a quadratic number of
  variables.
\newblock {\em Journal of Computer and System Sciences}, 78(1):151--163, 2012.
\newblock \href {https://doi.org/10.1016/j.jcss.2011.01.004}
  {\path{doi:10.1016/j.jcss.2011.01.004}}.

\bibitem{Hol16}
Petter Holme.
\newblock Temporal network structures controlling disease spreading.
\newblock {\em Physical Review E}, 94.2:022305, 2016.
\newblock \href {https://doi.org/10.1103/PhysRevE.94.022305}
  {\path{doi:10.1103/PhysRevE.94.022305}}.

\bibitem{kempe2002connectivity}
David Kempe, Jon Kleinberg, and Amit Kumar.
\newblock Connectivity and inference problems for temporal networks.
\newblock {\em Journal of Computer and System Sciences}, 64(4):820--842, 2002.
\newblock \href {https://doi.org/10.1006/jcss.2002.1829}
  {\path{doi:10.1006/jcss.2002.1829}}.

\bibitem{kermack1927contribution}
William~Ogilvy Kermack and Anderson~G. McKendrick.
\newblock A contribution to the mathematical theory of epidemics.
\newblock {\em Proceedings of the Royal Society of London. Series A:
  Mathematical and Physical Sciences}, 115(772):700--721, 192o7.
\newblock \href {https://doi.org/10.1098/rspa.1927.0118}
  {\path{doi:10.1098/rspa.1927.0118}}.

\bibitem{KMMNZ21}
Nina Klobas, George~B. Mertzios, Hendrik Molter, Rolf Niedermeier, and Philipp
  Zschoche.
\newblock Interference-free walks in time: Temporally disjoint paths.
\newblock In {\em Proceedings of the 30th International Joint Conference on
  Artificial Intelligence (IJCAI)}, pages 4090--4096. International Joint
  Conferences on Artificial Intelligence Organization, 2021.
\newblock \href {https://doi.org/10.24963/ijcai.2021/563}
  {\path{doi:10.24963/ijcai.2021/563}}.

\bibitem{DBLP:journals/jal/MahajanR99}
Meena Mahajan and Venkatesh Raman.
\newblock Parameterizing above guaranteed values: Maxsat and maxcut.
\newblock {\em Journal of Algorithms}, 31(2):335--354, 1999.
\newblock \href {https://doi.org/10.1006/jagm.1998.0996}
  {\path{doi:10.1006/jagm.1998.0996}}.

\bibitem{DBLP:journals/jcss/MahajanRS09}
Meena Mahajan, Venkatesh Raman, and Somnath Sikdar.
\newblock Parameterizing above or below guaranteed values.
\newblock {\em Journal of Computer and System Sciences}, 75(2):137--153, 2009.
\newblock \href {https://doi.org/10.1016/j.jcss.2008.08.004}
  {\path{doi:10.1016/j.jcss.2008.08.004}}.

\bibitem{mertzios2019temporal}
George~B Mertzios, Othon Michail, and Paul~G Spirakis.
\newblock Temporal network optimization subject to connectivity constraints.
\newblock {\em Algorithmica}, 81(4):1416--1449, 2019.
\newblock \href {https://doi.org/10.1007/s00453-018-0478-6}
  {\path{doi:10.1007/s00453-018-0478-6}}.

\bibitem{michail2016traveling}
Othon Michail and Paul~G. Spirakis.
\newblock Traveling salesman problems in temporal graphs.
\newblock {\em Theoretical Computer Science}, 634:1--23, 2016.
\newblock \href {https://doi.org/10.1016/j.tcs.2016.04.006}
  {\path{doi:10.1016/j.tcs.2016.04.006}}.

\bibitem{DBLP:books/daglib/0012859}
Michael Mitzenmacher and Eli Upfal.
\newblock {\em Probability and Computing: Randomized Algorithms and
  Probabilistic Analysis}.
\newblock Cambridge University Press, 2005.
\newblock \href {https://doi.org/10.1017/CBO9780511813603}
  {\path{doi:10.1017/CBO9780511813603}}.

\bibitem{MRZ21}
Hendrik Molter, Malte Renken, and Philipp Zschoche.
\newblock Temporal reachability minimization: Delaying vs. deleting.
\newblock In {\em Proceedings of the 46th International Symposium on
  Mathematical Foundations of Computer Science (MFCS)}, volume 202 of {\em
  Leibniz International Proceedings in Informatics}, pages 76:1--76:15. Schloss
  Dagstuhl -- Leibniz-Zentrum f{\"u}r Informatik, 2021.
\newblock \href {https://doi.org/10.4230/LIPIcs.MFCS.2021.76}
  {\path{doi:10.4230/LIPIcs.MFCS.2021.76}}.

\bibitem{New18}
Mark E~J Newman.
\newblock {\em Networks}.
\newblock Oxford University Press, 2018.

\bibitem{PS11}
Raj~Kumar Pan and Jari Saram\"aki.
\newblock Path lengths, correlations, and centrality in temporal networks.
\newblock {\em Physical Review E}, 84(1):016105, 2011.
\newblock \href {https://doi.org/10.1103/PhysRevE.84.016105}
  {\path{doi:10.1103/PhysRevE.84.016105}}.

\bibitem{thejaswi2020restless}
Suhas Thejaswi, Juho Lauri, and Aristides Gionis.
\newblock Restless reachability problems in temporal graphs.
\newblock {\em CoRR}, abs/2010.08423, 2020.

\bibitem{wu_efficient_2016}
Huanhuan Wu, James Cheng, Yiping Ke, Silu Huang, Yuzhen Huang, and Hejun Wu.
\newblock Efficient algorithms for temporal path computation.
\newblock {\em IEEE Transactions on Knowledge and Data Engineering},
  28(11):2927--2942, 2016.
\newblock \href {https://doi.org/10.1109/TKDE.2016.2594065}
  {\path{doi:10.1109/TKDE.2016.2594065}}.

\end{thebibliography}

\end{document}